%% file: main.tex
\documentclass[runningheads,oribibl]{llncs}
\usepackage[T1]{fontenc}
\usepackage[numbers,sort&compress]{natbib}

\makeatletter
\renewcommand\@biblabel[1]{#1.}
\makeatother

\usepackage[colorlinks]{hyperref}
\usepackage{color}

\input{defs}

\begin{document}
\title{Mixed Capability Games}
%
%
%
%
%
\author{Kai Jia \and Martin Rinard \and Yichen Yang}
\institute{Department of Electrical Engineering and Computer Science, \\
Massachusetts Institute of Technology, USA \\
\email{\{jiakai,rinard,yicheny\}@csail.mit.edu}}
%
\maketitle              

\input{abstract}

\input{allcontent}

%
%
%
\newpage
\bibliographystyle{splncs04nat}
\bibliography{refs}

\end{document}

%% file: defs.tex
\usepackage{amsmath}
\usepackage{amssymb}
\usepackage{bm}
\usepackage{bbm}
\usepackage{mathtools}
\usepackage{adjustbox}
\usepackage{float}
\usepackage{subcaption}
\usepackage[lighttt]{lmodern}
\usepackage[inline]{enumitem}
\usepackage[font=small]{caption}
\usepackage[nameinlink,capitalise]{cleveref}
\usepackage{algorithm}
\usepackage{algpseudocode}

\usepackage{tikz}
\usepackage{ifthen}
\usetikzlibrary{shapes.misc}
\usetikzlibrary{arrows.meta}
\usetikzlibrary{calc}
\usetikzlibrary{decorations.pathreplacing}
\tikzset{cross/.style={cross out, draw=black, minimum size=2*(#1-\pgflinewidth),
    inner sep=0pt, outer sep=0pt},
cross/.default={1pt}
}

\newcommand{\real}{\mathbb{R}}

\newcommand{\posint}{\mathbb{Z}^+}
\newcommand{\intint}[2]{[#1,\,#2]}
\newcommand{\indicator}[1]{\mathbbm{1}_{#1}}

\newcommand{\defeq}{\mathrel{\raisebox{-0.3ex}{$\overset{\text{\tiny def}}{=}$}}}
\newcommand{\condSet}[2]{\left\{#1 \;\middle|\; #2\right\}}
\newcommand{\condExp}[2]{\operatorname{E}\left[#1 \;\middle|\; #2\right]}

\newcommand{\captrfn}{capability transfer function}

\newcommand{\CapTrFn}{Capability Transfer Function}
\newcommand{\agmg}{MGMG}

\newcommand{\capA}{C_{\operatorname{A}}}
\newcommand{\capB}{C_{\operatorname{B}}}
\newcommand{\fA}{f_{\operatorname{A}}}
\newcommand{\fB}{f_{\operatorname{B}}}

\newcommand{\fAp}{f_{\operatorname{A}}'}

\newcommand{\fBp}{f_{\operatorname{B}}'}
\DeclareMathOperator{\seg}{Seg} 
\newcommand{\stgLvl}[1]{\mathcal{L}_{#1}} 
\newcommand{\stistgLvl}[1]{\widetilde{\mathcal{L}}_{#1}} 
\newcommand{\payAsym}{u_{\operatorname{A}}}
\newcommand{\payBsym}{u_{\operatorname{B}}}
\newcommand{\payA}[2]{\payAsym\left(#1,\, #2\right)}
\newcommand{\payB}[2]{\payBsym\left(#1,\, #2\right)}

\DeclareMathOperator{\setN}{\mathcal{N}}  
\DeclareMathOperator{\setg}{\mathcal{G}}  
\DeclareMathOperator{\setm}{\mathcal{M}}  
\DeclareMathOperator{\nrg}{\prescript{\#}{}{g}}  
\DeclareMathOperator{\nrm}{\prescript{\#}{}{m}}  
\newcommand{\idiv}[2]{\left\lfloor\frac{#1}{#2}\right\rfloor} 
\DeclareMathOperator{\discup}{\mathcal{D}^{\uparrow}}
\DeclareMathOperator{\discdown}{\mathcal{D}^{\downarrow}}
\newcommand{\card}[1]{\left\lvert#1\right\rvert} 

\newcommand{\V}[1]{\bm{#1}}
\newcommand{\vc}{\V{c}}
\newcommand{\vs}{\V{s}}

\newcommand{\RNum}[1]{\uppercase\expandafter{\romannumeral #1\relax}}

\newcommand{\X}{,\,\allowbreak}

\newcommand{\algotext}[2]{\State \Comment{
    \parbox[t]{.93\linewidth-(3ex*#1)}{#2}}
}

\newenvironment{enuminline}
{\begin{enumerate*}[label=(\roman*),itemjoin={{; }},itemjoin*={{; and }}]}
{\end{enumerate*}}


%% file: abstract.tex
\begin{abstract}
    We present a new class of strategic games, \emph{mixed capability games}, as
    a foundation for studying how different player capabilities impact the
    dynamics and outcomes of strategic games. We analyze the impact of different
    player capabilities via a \emph{capability transfer function} that
    characterizes the payoff of each player at equilibrium given capabilities
    for all players in the game. In this paper, we model a player's capability
    as the size of the strategy space available to that player. We analyze a
    mixed capability variant of the Gold and Mines Game recently proposed by
    \citet{ yang2022impact} and derive its \captrfn{} in closed form.
\end{abstract}


%% file: allcontent.tex
\input{introduction}
\input{mixed-capability-games}
\input{gold-and-mine}

%% file: introduction.tex
\section{Introduction}

Player capabilities can significantly impact the dynamics and outcomes of
strategic games. Recently, \citet{yang2022impact} analyzed how different player
capabilities affect the social welfare in several congestion games. The research
models player strategies as programs in a domain-specific language and models
the capability of each player as the size of the programs available to that
player. All players in a given game have the same capability, with player
capabilities varying across games but not within the same game.

We present \emph{mixed capability games} as a general framework for studying
games in which players have different capabilities, both within the same game
and across different games. To capture how game outcomes depend on different
player capabilities, we propose analyzing a \emph{\captrfn} that precisely
quantifies the payoffs of individual players given the capabilities of all
players in a game. \Cref{sec:mcg} presents the concepts in our framework.
\Cref{sec:mgmg} presents an analysis of a mixed capability game, the Mixed Gold
and Mines Game, and derives closed-form expressions for the \captrfn{} of this
game.


%% file: mixed-capability-games.tex
\section{Mixed Capability Games and \CapTrFn\label{sec:mcg}}

We model the capability of each player as the size of the strategy space
available to that player. We first present formal definitions for pure Nash
equilibria of normal-form games \citep{nisan2007algorithmic}, then extend the
definitions to mixed Nash equilibria.

\begin{definition}
    A \emph{mixed capability game} is a tuple $G = (\setN \X
        (b_i)_{i\in\setN} \X
        (\stgLvl{j}^i)_{i\in\setN,\, 1\leq j \leq b_i} \X
        (u_i)_{i\in\setN}
    )$ where:
    \begin{itemize}
        \item $\setN = \{1,\ldots,n\}$ is the set of players.
        \item $b_i \in \posint$ is the maximal capability of player $i$.
        \item $\stgLvl{j}^i$ is the strategy space of player $i$ when they have
            capability $j$. We also require that the strategy spaces of a player
            form a hierarchy: $\forall 1 \leq j < b_i:\: \stgLvl{j}^i \subsetneq
            \stgLvl{j+1}^i$, i.e., a player has more strategies to choose from
            when they have higher capability.
        \item $u_i: \stgLvl{b_1}^1 \times \cdots \times \stgLvl{b_n}^n \mapsto
            \real$ is the payoff function that computes the payoff for player
            $i$ given the strategies chosen by all players.
    \end{itemize}
\end{definition}

A specification of the actual capabilities of players is necessary to determine
the outcome of the game.

\begin{definition}
    A \emph{capability profile} for a mixed capability game is a tuple of
    integers $\vc = (c_1, \ldots, c_n) $ where $1 \leq c_i \leq b_i$. A
    capability profile determines the strategy spaces of the players. Player $i$
    can choose strategies only from $\stgLvl{c_i}^i$.
\end{definition}

Given a capability profile $\vc = (c_1, \ldots, c_n)$, a \emph{strategy profile}
is a tuple $\vs = (s_1, \ldots, s_n)$ where $s_i \in \stgLvl{c_i}^i$ that
specifies the strategies chosen by all players. A strategy profile is a
\emph{pure Nash equilibrium} if no player can improve their payoff by
unilaterally changing their strategy: $\forall 1 \leq i \leq n:\: u_i(\vs) =
\max_{s_i' \in \stgLvl{c_i}^i} u_i(s_i',\,\vs_{-i}) $. The notation $(s_i',\,
\vs_{-i})$ denotes a new strategy profile in which player $i$ plays strategy
$s_i'$ and all other players play the same strategy as in $\vs$.

\begin{definition}
    A \captrfn{} of a mixed capability game is a function $f: \intint{1}{b_1}
    \times \cdots \times \intint{1}{b_n} \mapsto 2^{\real^n}$ where $\intint{a}
    {b}$ denotes the integers between $a$ and $b$, and $2^S$ is the power set of
    a set $S$. The \captrfn{} computes the set of player payoffs at equilibrium
    for a capability profile. Formally, given a capability profile $\vc = (c_1,
    \ldots, c_n)$, $f(\vc)$ is a set such that $(y_1, \ldots, y_n) \in f(\vc)$
    if and only if there is a pure Nash equilibrium $\vs = (s_1, \ldots, s_n)$
    for which $s_i \in \stgLvl{c_i}^i$ and $y_i = u_i(\vs)$.
\end{definition}

The \captrfn{} contains detailed information about the game's behavior under
varying player capabilities. \Cref{eg:captrfn:cap-pos} illustrates how to use a
\captrfn{} to define the  higher level concept of a \emph{capability-positive
game}.

\begin{example}
    \label{eg:captrfn:cap-pos}
    Capability-positive games \citep{yang2022impact} are games in which
    \begin{enuminline}
        \item all players share the same capability
        \item social welfare at equilibrium cannot decrease as players become
            more capable
    \end{enuminline}.
    Such games can be defined using the \captrfn{} for that game. A game is capability-positive if
    $\max W_b \leq \min W_{b+1}$ where $W_b \defeq \condSet{\sum_{j\in i}j}{i
    \in f(b, \ldots,b)}$. Note that $W_b$ is the set of social welfare at
    equilibrium defined via the \captrfn{} of this game.
\end{example}

We extend the definitions to games without pure Nash equilibria. We consider
\emph{mixed Nash equilibria} in which players act stochastically. All finite
games have mixed Nash equilibria\citep{nash1950equilibrium}. Given a capability
profile $\vc = (c_1, \ldots, c_n)$, the strategy of a player $i$ is a
distribution over possible actions, denoted as $P(a|s_i)$ where $a \in
\stgLvl{c_i}^i$. Player $i$ receives expected payoff $\condExp{u_i}{\vs}$:
\begin{align*}
    \condExp{u_i}{\vs} = \sum_{a_j \in \stgLvl{c_j}^j} u_i(a_1,\,\ldots,\,a_n)
        P(a_1|s_1) \cdots P(a_n|s_n)
\end{align*}
A strategy profile is a mixed Nash equilibrium if no player can unilaterally
change their own distribution to improve their expected payoff. In this case,
the \captrfn{} is defined as the set of expected payoffs of all mixed Nash
equilibria given a capability profile.

\begin{definition}
    The \captrfn{} of a mixed capability game with mixed Nash equilibria is a
    function $f: \intint{1}{b_1} \times \cdots \times \intint{1}{b_n} \mapsto
    2^{\real^n}$. Given a capability profile $\vc = (c_1,
    \ldots, c_n)$, $f(\vc)$ is a set such that $(y_1, \ldots, y_n) \in f(\vc)$
    if and only if there is a mixed Nash equilibrium $\vs = (s_1, \ldots, s_n)$
    for which $s_i$ defines a distribution over $\stgLvl{c_i}^i$ and $y_i =
    \condExp{u_i}{\vs}$.
\end{definition}

One natural question regarding mixed capability games is whether increasing the
capability for one player does not make this player receive less payoff.
Formally, let $f_i(\vc) \defeq \condSet{y_i}{(y_1, \ldots, y_n) \in f(\vc)}$
denote the set of payoffs of player $i$ at equilibrium, then the question is
whether $\min f_i(\vc) \leq \max f_i(\vc')$ for each $i$ where $\vc = (c_1,
\ldots, c_i, \ldots, c_n)$ and $\vc' = (c_1, \ldots, c_i', \ldots, c_n)$ with
$c_i' > c_i$. \Cref{eg:captrfn:decr} shows that this is not necessarily true for
Nash equilibria since the player with increased capability may switch to another
strategy, which triggers responses of other players that ultimately reduce the
payoff of the initial player. Note that in a Stackelberg game \citep{
simaan1973stackelberg} where the leader announces their strategy before others
simultaneously choose their responses, the \captrfn{} is monotonic for the
leader.

\begin{example}
    \label{eg:captrfn:decr}
    Consider a two-player two-action bimatrix game. Player 1 is the row player
    with two possible capabilities: $\stgLvl{1}^1 = \{1\}$ and $\stgLvl{2}^1 =
    \{1, 2\}$. Player 2, the column player, has one capability: $\stgLvl{1}^2 =
    \{1, 2\}$. Their payoff matrices are:
    \begin{align*}
        u_1 =
        \begin{pmatrix}
            1 \; & -1 \\
            2 \; & 0
        \end{pmatrix} \hspace{3em}
        u_2 =
        \begin{pmatrix}
            2 \; & 1 \\
            1 \; & 2
        \end{pmatrix}
    \end{align*}
    When player 1 is has capability 1, they can only play the first row, and
    player 2 plays the first column, which gives payoffs 1 and 2 for each player
    respectively. Therefore, we have $f(1,\,1) = \{(1,\, 2)\}$ for the \captrfn.
    When player 1 is allowed to use full capability, the only Nash equilibrium
    is (second row, second column), which gives $f(2,\,1) = \{(0,\, 2)\}$. The
    \captrfn{} decreases for player 1 even though player 1's capability
    increases.
\end{example}


%% file: gold-and-mine.tex
\section{Mixed Gold and Mines Game\label{sec:mgmg}}

\begin{figure}[t]
    \centering
    \begin{adjustbox}{width=0.6\textwidth}
        \input{fig/agmg-illu.tex}
    \end{adjustbox}
    \caption{
        An example \agmg{} instance. Each dot (resp. cross) is a gold (resp.
        mine). The dashed lines represent a PNE when $\capA=1$ and $\capB=2$
        (with $\rho < -\mu < 1$).
        \label{fig:agmg-illu}
    }
    \vskip-.5em
\end{figure}
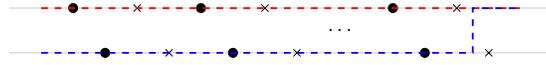

We derive exact expressions for the \captrfn{} of an asymmetric version of the
alternating ordering Gold and Mines Game, a special case of distance-bounded
network congestion games originally proposed by \citet{ yang2022impact}. We name
this new game the Mixed Gold and Mines Game (\agmg). Unlike the previous Gold
and Mines Game of \citet{yang2022impact} in which all players in the same game
have the same capability, in a single Mixed Gold and Mine Game, players may have
different capabilities.

\agmg{} is a two-player congestion game parameterized by five numbers $(
M\in\posint \X \rho\in\real \X \mu\in\real \X \capA\in\posint \X \capB\in\posint)
$. \agmg{} has resources arranged as a specific pattern; players use line
segments to cover resources to receive payoffs. As in all congestion games,
\agmg{} games always have pure Nash equilibria~\citep{rosenthal1973class}.

\paragraph{Resource Layout:} Each \agmg{} game has $4M$ resources arranged on
two lines. Each resource is either a gold site or a mine site. Each line
contains $M$ gold sites and $M$ mine sites in alternating order. Resources are
placed at distinct horizontal locations $0, 1, \cdots, 4M-1$. For the resource
at location $i$, $y_i = (i + 1) \bmod 2$ indicates which line it is placed on,
and $t_i = \indicator{i \bmod 4 \,\leq\, 1}$ indicates whether it is a gold site
($t_i=1$) or a mine site ($t_i=0$).

\paragraph{Game Objective:} Two players maximize their payoff by using line
segments to cover the resources. A \emph{delay function} $r_t(n)$, where
$t\in\{m, g\}$ is the resource type and $n \in \{1,\,2\}$ is the number of
players covering the resource, specifies the payoff for covering a resource. For
gold sites, $r_g(1) = 1$ and $r_g(2)=\rho$, where $0 < \rho < 1$ so that if two
players cover a resource they receive a smaller payoff. For mine sites, $r_m(n)
= \mu < 0$ is a constant penalty.

\paragraph{Strategy Space:} Each player $p \in \{A,\,B\}$ uses a function $f_p:
[0,\, 4M-1]\mapsto \{0,\,1\}$ to specify which line player $p$ covers at each
horizontal location. Player $p$ covers the resource at location $i$ if $f_p(i) =
y_i = (i+1)\bmod 2$. The strategy space $\stgLvl{C}$ of a player with capability
$C$ contains all functions with no more than $C$ segments:
\begin{align*}
    \seg(f) &\defeq
        \left(\sum_{i=0}^{4M-2}\indicator{f(i)\neq f(i+1)}\right) + 1 \\
    \stgLvl{C} &\defeq \condSet{f: [0,\, 4M-1]\mapsto \{0,\,1\}}{\seg(f)
        \leq C}
\end{align*}
We use $\capA$ and $\capB$ to denote the capabilities of player A and B
respectively, and use $\fA(\cdot)$ and $\fB(\cdot) $ for their strategies. Note
that \citet{yang2022impact} shows that the capability bound has a  natural
interpretation as the size of programs in a Domain-Specific Language (DSL)
describing the strategy space.

Below is our main result:
\begin{theorem}
    \label{thm:agmg}
    Given an instance of \agmg{} parameterized by $(M \X \rho \X \mu \X \capA \X
    \capB)$ that satisfies $0 < \rho < -\mu < 1$, in a pure Nash equilibrium of
    this game, the players receive the following payoffs $\payAsym$ and
    $\payBsym$:
    \begin{align*}
        \payAsym &= \idiv{\capA'+t-1}{2}\rho - \idiv{\capA' - t}{2}\mu
                +\idiv{\capB'-t}{2}(\rho-1) + (\mu+1)M \\
        \payBsym &= \idiv{\capB'-t}{2}\rho -  \idiv{\capB'+t-1}{2}\mu
                +\idiv{\capA'+t-1}{2}(\rho-1) + (\mu+1)M \\
        \text{where} & \\
        \capA' &= \min(\capA,\, 2M+1), \hspace{1em}
        \capB' = \min(\capB,\, 2M+1), \hspace{1em}
        t \in \{0,\,1\}
    \end{align*}
    Three cases determine the value of $t$:
    \begin{itemize}
        \item When $\max(\capA,\,\capB) \leq 2M$, there are two classes of Nash
            equilibria distinguished by $t=0$ and $t=1$.
        \item When $\min(\capA,\, \capB) \leq 2M < \max(\capA,\,\capB)$, $t=0$
            if $\capA \leq 2M$ and $t=1$ if $\capB \leq 2M$.
        \item When $\min(\capA,\,\capB) \geq 2M+1$, there is one Nash
            equilibrium. The above formulas give the same payoffs regardless of
            $t=0$ or $t=1$.
    \end{itemize}
\end{theorem}

In \agmg{}, when one player's capability increases, their own payoff increases
by $\rho$ or $-\mu$, but their opponent's payoff decreases by $\rho - 1$. If
both players get the same capability increment, the social welfare (i.e., the
sum of their payoffs) can increase, decrease, or stay the same, depending on the
sign of $2\rho-\mu-1$. If both players have the same capability $C$, the social
welfare is $\payAsym + \payBsym = (2\rho-\mu-1) (\min(C,\,2M+1)-1) + 2(\mu+1)M$,
which confirms Theorem 16 of \citet{ yang2022impact} up to a constant bias
because in \agmg{} we remove the last gold site on each line to simplify our
analysis.

In the following two sections, we first present three lemmas that characterize the
Nash equilibria in \agmg{}, and then derive the above results based on these
lemmas.

\subsection{Characteristics of Nash equilibria}

We introduce some notation:
\begin{itemize}
    \item A pair $(\fA,\, \fB)$ denotes a strategy profile, i.e., the
        strategies of both players.
    \item Given a strategy profile, $\payA{\fA}{\fB}$ and $\payB{\fA}{\fB}$ are
        the payoffs of individual players.
    \item Given a strategy $f(\cdot)$, $\setg(f)/\setm(f)$ and
        $\nrg(f)/\nrm(f)$ denote the locations and numbers of gold and mine
        sites covered by the strategy:
        \begin{align*}
            \setg(f)  &\defeq \condSet{4i}{0 \leq i < M \text{ and } f(4i)=1} \\
                & \hspace{2em}\cup
                \condSet{4i+1}{0 \leq i < M \text{ and } f(4i+1)=0} \\
            \setm(f)  &\defeq
                \condSet{4i+2}{0 \leq i < M \text{ and } f(4i+2)=1} \\
                & \hspace{2em}\cup
                \condSet{4i+3}{0 \leq i < M \text{ and } f(4i+3)=0} \\
            \nrg(f) &\defeq \card{\setg(f)} \\
            \nrm(f) &\defeq \card{\setm(f)}
        \end{align*}
    \item Given a strategy $f_p(\cdot)$ for player $p$, \emph{discontinuity
        points} (DPs) are the locations where $f_p(\cdot)$ changes the line that
        $p$ covers. We also differentiate between \emph{upward discontinuity
        points} (UDPs, denoted by $\discup(f)$) and \emph{downward
        discontinuity points} (DDPs, denoted by $\discdown(f)$):
        \begin{align*}
            \discup(f) &\defeq \condSet{i}{0 \leq i \leq 4M - 2 \text{ and }
            f(i) = 0 \text{ and } f(i+1) = 1} \\
            \discdown(f) &\defeq \condSet{i}{0 \leq i \leq 4M - 2 \text{ and }
            f(i) = 1 \text{ and } f(i+1) = 0}
        \end{align*}
        Note that $\seg(f) = \card{\discup(f)} + \card{\discdown(f)} + 1$.
    \item A strategy $f(\cdot)$ is a \emph{perfect cover} for resources located
        between $[a,\, b]$ if all gold sites are covered and all mine sites are
        avoided: $[a,\,b] \setminus \setg(f) = \emptyset$ and $[a,\, b] \cap
        \setm(f) = \emptyset$. We also call the resources $[a,\,b]$
        \emph{perfectly covered} in this case, and \emph{imperfectly covered}
        otherwise. Here $[a,\,b]$ denotes all integers in the interval: $[a,\,b]
        = \condSet{i}{a \leq i \leq b}$. Note that to perfectly cover resources
        $[4i,\,4j-1]$ for $i < j$, one needs $2(j-i)+1$ segments.
    \item \emph{Strict strategy spaces} use exactly the given number of segments:
        \begin{align*}
            \stistgLvl{1} &\defeq \stgLvl{1} \\
            \stistgLvl{C} &\defeq \stgLvl{C} - \stgLvl{C - 1} \\
                &= \condSet{f: [0,\, 4M-1]\mapsto \{0,\,1\}}{\seg(f) = C}
        \end{align*}
    \item A strategy profile $(\fA,\, \fB)$ is a \emph{complete-gold-coverage}
        for a \agmg{} if both players cover all gold sites together, i.e.,
        $\card{\setg(\fA) \cup \setg(\fB)} = 2M$.
\end{itemize}

First, we show that DPs only occur at certain locations:
\begin{lemma}
    \label{lemma:agmg:disc-loc}
    Let $f(\cdot)$ be a best response of a player given the other player's
    strategy. Upward discontinuity points in $f(\cdot)$ occur only at
    neighboring mine sites, and downward discontinuity points occur only at
    neighboring gold sites:
    \begin{align*}
        \forall i \in \discup(f) &: i \bmod 4 = 2 \\
        \forall i \in \discdown(f) &: i \bmod 4 = 0 \\
    \end{align*}
\end{lemma}
\begin{proof} 
    \newcommand{\drawfig}[1]{
        \centering
        \begin{adjustbox}{width=.98\textwidth}
            \begin{tikzpicture}
                \input{fig/agmg-baseunit.tex}
                \draw[line] (\xsep * -0.3, 0) --
                    ({\xsep * #1}, 0) -- ({\xsep * #1}, \ysep) --
                    (\xsep * 4.3, \ysep);
            \end{tikzpicture}
        \end{adjustbox}
    }
    \newcommand{\addsubfig}[2]{
        \begin{subfigure}[t]{.23\textwidth}
            
        \centering
        \begin{adjustbox}{width=.98\textwidth}
            \begin{tikzpicture}
                \input{fig/agmg-baseunit.tex}
                \draw[line] (\xsep * -0.3, 0) --
                    ({\xsep * (#1+0.5)}, 0) -- ({\xsep * (#1+0.5)}, \ysep) --
                    (\xsep * 4.3, \ysep);
            \end{tikzpicture}
        \end{adjustbox}
    
            \vskip-1.5em
            \caption{Upward at #1. Payoff is #2.}
        \end{subfigure}
    }
    \begin{figure}[tb]
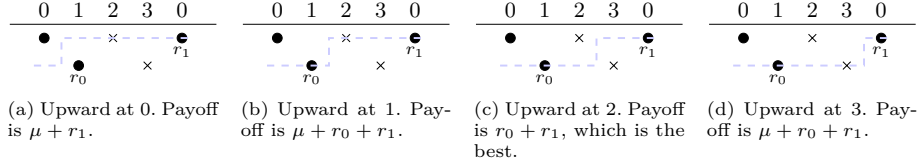

        \centering
        \addsubfig{0}{$\mu+r_1$}
        \hfill
        \addsubfig{1}{$\mu+r_0+r_1$}
        \hfill
        \addsubfig{2}{$r_0+r_1$, which is the best}
        \hfill
        \addsubfig{3}{$\mu + r_0+r_1$}
        \vskip-.5em
        \caption{
            Cases of a $f(4k)=0$ and $f(4k+4)=1$ in a local region with one DP.
            The numbers are locations of resources modulo 4. Dashed lines
            indicate a local part of the strategy.
            \label{fig:agmg:disc-loc-proof}
        }
    \end{figure}
    We consider cases in a local region for different values of $f(4k)$ and
    $f(4k+4)$.

    \begin{itemize}
        \item $f(4k)=0,\,f(4k+4)=1$:
            \Cref{fig:agmg:disc-loc-proof} shows the cases with one DP. The
            payoffs of covered gold sites are denoted as $r_0$ and $r_1$, which
            can be $1$ or $\rho$ depending on the opponent's strategy. Clearly,
            the payoff is maximized only when the DP is at location 2 modulo 4.
            It can be verified that using more DPs while maintaining $f(4k)=0$
            and $f(4k+4)=1$ does not improve payoff.
        \item $f(4k)=1,\,f(4k+4)=0$:
            Similarly, the best response in this case has one DDP at location 0.
        \item $f(4k)=f(4k+4)=0$:
            The best response should have no DP. If there are DPs, there should
            be one UDP and one DDP to cover one gold site and no mine site, but
            moving the DDP rightward to also cover the gold at $4k+4$ gives
            better payoff with the same number of segments.
        \item $f(4k)=f(4k+4)=1$:
            The best response should either have no DP (covering two gold sites
            and one mine site) or two DPs (covering three gold sites and no mine
            site) at locations 0 and 2 modulo 4.
    \end{itemize}
\end{proof} 

Now we show that the number of gold and mine sites covered by an optimal
strategy is fairly predictable, i.e., it only depends on $f(0)$ and $\seg(f)$:
\begin{lemma}
    \label{lemma:agmg:gm-cnt}
    If $f(\cdot)$ is a strategy that conforms to \cref{lemma:agmg:disc-loc}, then
    \begin{align*}
        \nrg(f) &= M + \idiv{\seg(f) + f(0) - 1}{2} \\
        \nrm(f) &= M - \idiv{\seg(f) - f(0)}{2} \\
    \end{align*}
\end{lemma}
\begin{proof} 
We define a series of strategies $\{f_i\}$. Let $f_0 = f$ and define $f_{i+1}$
the strategy obtained by removing the last DP of $f_i$, i.e., $f_{i+1}(x) =
f_i(\min(x,\, x_i))$ where $x_i = \max \left( \discup(f_i) \cup \discdown(f_i)
\right)$. \Cref{lemma:agmg:disc-loc} implies that each DDP adds an extra gold
site and each UDP avoids a mine site, which means either $\nrg(f_{i}) -
\nrg(f_{i+1}) = 1$ or $\nrm(f_{i+1}) - \nrm(f_{i}) = 1$, depending on whether the
last DP of $f_i$ is DDP or UDP. It follows that $\nrg(f_i) = \nrg(f_j) +
\card{\discdown(f_i)} - \card{\discdown(f_j)}$ and $\nrm(f_i) = \nrm(f_j) -
\card{\discup(f_i)} + \card{\discup(f_j)}$ for any pair $i,\, j$.

We first assume $f(0)=1$. In this case, since UDPs and DDPs are interleaving, we have
$\card{\discdown(f)} = \idiv{\seg(f)}{2}$ and $\card{ \discup(f)} = \idiv{\seg(f)
-1}{2}$. Let $c=\seg(f)$. We also know that $\nrg(f_c) = \nrm(f_c)=M$ since
$f_c$ covers exactly one line. Therefore, $\nrg(f) = \nrg(f_0) = \nrg(f_c) +
\card{\discdown(f_0)} - \card{\discdown(f_c)} =  M + \idiv{\seg(f)}{2}$ and
$\nrm(f) = M-\idiv{\seg(f) - 1}{2}$. A similar analysis
for the case $f(0)=0$ gives $\nrg(f) = M + \idiv{\seg(f)-1}{2}$ and $\nrm(f)
= M - \idiv{\seg(f)}{2}$. \Cref{lemma:agmg:gm-cnt} summarizes these results
compactly.
\end{proof} 

Now let's shift our attention from one player's strategy to a strategy profile.
\begin{lemma}
    \label{lemma:agmg:cover-all-gold}
    If  $\rho < -\mu < 1$ and $(\fA, \, \fB)$ is a pure Nash equilibrium when
    players are limited to the strict strategy spaces $\stistgLvl{\capA}$ and
    $\stistgLvl{\capB}$, then $(\fA, \, \fB)$ is a complete-gold-coverage.
\end{lemma}
\begin{proof} 
    We prove this statement in two steps. We first show that for any player,
    their payoff is maximized when they cover as many unoccupied gold sites as
    possible. Then we show that complete-gold-coverage is always feasible.

    Let $T(\fA,\,\fB)$ be the total number of gold sites covered by a strategy
    profile: $T(\fA,\, \fB) \defeq \card{\setg(\fA) \cup \setg(\fB)}$.

    Without loss of generality, we focus on player A. We show that if there is a
    strategy $\fAp$ such that $T(\fA,\,\fB) < T(\fAp,\, \fB)$ where $\{\fA,\,
    \fAp\} \subset \stistgLvl{\capA}$, then $(\fA,\,\fB)$ is not a Nash
    equilibrium because A can get better payoff by switching to $\fAp$. Note
    that the number of gold sites covered by both players in the strategy
    profile $(\fA,\, \fB)$ is $\nrg(\fA) + \nrg(\fB) - T(\fA,\,\fB)$, while the
    number of gold sites covered by A exclusively is $T(\fA,\,\fB) - \nrg(\fB)$,
    which implies:
    \begin{align*}
        \payA{\fA}{\fB} &=
            \left(T(\fA,\,\fB) - \nrg(\fB)\right) \cdot 1 \\
            &\hspace{2em} +
            \left( \nrg(\fA) + \nrg(\fB) - T(\fA,\,\fB) \right)
                \cdot \rho +
            \nrm(\fA) \cdot \mu
    \end{align*}
    Substituting the results of \cref{lemma:agmg:gm-cnt} into the above:
    \begin{align*}
        \payA{\fA}{\fB}
        &= \rho \idiv{\capA + \fA(0) - 1}{2} +
            (-\mu) \idiv{\capA - \fA(0)}{2} +
            (1-\rho)T(\fA,\,\fB) \\
            & \hspace{2em} + (\rho-1)\nrg(\fB) + (\rho+\mu)M
    \end{align*}
    Let $h(f) \defeq \rho\idiv{\capA + f(0) - 1}{2} + (-\mu) \idiv{\capA - f(0)}
    {2}$ be the first two terms. One can verify that $h(\fAp) - h(\fA) \in
    \{0,\, \rho + \mu,\, -\rho-\mu\}$ for all possible values of $\capA$, $\fA(0)
    $, and $\fAp(0)$. Note $\rho < -\mu$ implies $\rho+\mu<0$. Thus $h(\fAp) -
    h(\fA) \geq \rho + \mu$.
    \begin{align*}
        \payA{\fAp}{\fB} - \payA{\fA}{\fB}
        &= h(\fAp) - h(\fA) +
            \left(T(\fAp,\, \fB) - T(\fA,\, \fB) \right)(1-\rho) \\
        &\geq (\rho + \mu) + (1 - \rho) \\
        & > 0
    \end{align*}

    Therefore, A first maximizes $T(\fA,\,\fB)$ and then maximizes $h(\fA)$ in
    their best response. The maximum possible value of $T(\fA,\,\fB)$ is $2M$
    which is achieved when $(\fA,\,\fB)$ is a complete-gold-coverage.

    Next we show that complete-gold-coverage is always feasible. We assume
    $\capA \geq \capB$ WLOG. For any strategy $\fB(\cdot)$ played by player B
    that conforms to \cref{lemma:agmg:disc-loc}, we show that there exists $\fA
    \in \stistgLvl{\capA}$ such that $T(\fA,\, \fB) = 2M$.

    If $\capA \geq 2M$, then A can cover all gold sites trivially. Now we
    consider the case $1 \leq \capA \leq 2M-1$. We first construct a strategy
    $\fAp(\cdot)$ that may or may not use all the segments. For $0 \leq k < M$,
    we set $\fAp(4k) = 1-\fB(4k)$ and $\fAp(4k+3) = 1 - \fB(4k+3)$. When
    $\fAp(4k) = \fAp(4k+3)$, we use the same value for $\fAp(4k+1)$ and
    $\fAp(4k+2)$; otherwise we add one discontinuity point at $\fAp(4k)$ or
    $\fAp(4k+2)$ according to \cref{lemma:agmg:disc-loc}. Note that
    \cref{lemma:agmg:disc-loc} also implies $\fAp(4k-1) = \fAp(4k)$. It is easy
    to verify that $T(\fAp,\, \fB) = 2M$ and $\seg(\fAp) \leq \seg(\fB) \leq
    \capA$. We then derive $\fA$ from $\fAp$ using \cref{algo:moreseg} so that
    $\seg(\fA) = \capA$.

    \begin{algorithm}
        \caption{Modify a strategy to use more segments}
        \label{algo:moreseg}
        \begin{algorithmic}[1]
            \Require Game scale $M \geq 2$
            \Require Player capability $\capA$ such that $\capA \leq 2M-1$
            \Require A strategy $\fAp(\cdot)$ that conforms to
                \cref{lemma:agmg:disc-loc} such that $\seg(\fAp) \leq \capA$ and
                the last four resources are imperfectly covered.
            \Ensure A strategy $\fA(\cdot)$ that conforms to
                \cref{lemma:agmg:disc-loc} such that $\seg(\fA) = \capA$,
                $\setg(\fAp) \subseteq \setg(\fA)$, and $\fA(0) = \fAp(0)$.
            \State $k \gets 0$
            \State $\fA \gets \fAp$
            \While {$\capA - \seg(\fA) \geq 2$}
                \algotext{1}{
                    When entering the loop, all resources in $[4,\, 4k-1]$ are
                    perfectly covered and $\fA(4k)=1$ when $k \geq 1$. Perfectly
                    covering $[4,\,4k-1]$ requires $\seg(\fA) \geq 2(k-1) + 1$.
                    We also have $\seg(\fA) \leq \capA - 2 \leq 2M - 3$ due to
                    the loop condition, thus $2k-1 \leq \seg(\fA) \leq 2M-3$
                    which means $k \leq M - 1$. When $k \geq 1$, each iteration
                    modifies $\fA(\cdot)$ to perfectly cover $[4k,\,4k+3]$ using
                    no more than two new segments.} \label{algo:moreseg:loopinv}
                \If {$\fA(4k+3) = 0$}
                    \algotext{2}{
                        \cref{lemma:agmg:disc-loc} ensures $\fA(4k+i) = 0$
                        for $2 \leq i \leq 5$.}
                    \State $\fA(4k+3) \gets 1$
                    \If {$k + 1 < M$}
                        \State $\fA(4k+4) \gets 1$
                    \EndIf
                \ElsIf {$k > 0$ or $\fA(4k)=1$}
                    \algotext{2}{We now have $\fA(4k)=\fA(4k+3)=1$.}
                    \State $\fA(4k+1) \gets 0$
                    \State $\fA(4k+2) \gets 0$
                \EndIf
                \State $k \gets k + 1$
            \EndWhile
            \If {$\capA - \seg(\fA) = 1$}
                \algotext{1}{We have $\seg(\fA) = \capA - 1 < 2M-1$ in this
                    case. Thus the resources $[4, 4M-1]$ are imperfectly
                    covered. Due to our requirement on $\fAp(\cdot)$ and the way
                    we construct $\fA(\cdot)$, the last four resources are
                    imperfectly covered. We modify the strategy on the last few
                    resources to use one more segment.}
                \If {$\fA(4M-1) = 0$}
                    \State $\fA(4M-1) \gets 1$
                \ElsIf {$\fA(4M-2)=1$}
                    \algotext{2}{\cref{lemma:agmg:disc-loc} implies $\fA(4M-i) =
                        1$ for $1 \leq i \leq 4$}
                    \State $\fA(4M-3) \gets 0$
                    \State $\fA(4M-2) \gets 0$
                    \State $\fA(4M-1) \gets 0$
                \Else
                    \algotext{2}{Now $\fA(4M-2)=0$ and $\fA(4M-1)=1$. Since
                        the last four resources are imperfectly covered, we have
                        $\fA(4M-i)=0$ for $2\leq i \leq 6$.}
                    \State $\fA(4M-5) \gets 1$
                    \State $\fA(4M-4) \gets 1$
                    \State $\fA(4M-1) \gets 0$
                \EndIf
            \EndIf
            \State \Return $\fA$
        \end{algorithmic}
    \end{algorithm}

\end{proof} 

\subsection{\CapTrFn{} of \agmg}

Recall that in the proof of \cref{lemma:agmg:cover-all-gold}, we have shown that
given a strategy profile $(\fA,\,\fB)$ that is a complete-gold-coverage, A's
payoff is
\begin{align*}
    \payA{\fA}{\fB} &= \rho \idiv{\seg(\fA) + \fA(0) - 1}{2} +
        (-\mu)\idiv{\seg(\fA) - \fA(0)}{2} \\
        &\hspace{2em} + (\rho-1)\nrg(\fB) + (\mu - \rho + 2)M
\end{align*}

For two strategy profiles $(\fA,\,\fB)$ and $(\fAp,\, \fB)$ that are both
complete-gold-coverage, we make the following two observations that can be
verified using the above expansion of $\payA{\fA}{\fB}$:
\begin{enumerate}
    \item If $\seg(\fA) < \seg(\fAp)$ and $\min(-\mu,\,\rho)>0$, then $\payA{\fA}
        {\fB} < \payA{\fAp}{\fB}$.
    \item If $\fA(0)=1$, $\fAp(0)=0$, $\seg(\fA)=\seg(\fAp)$, and
        $\rho + \mu < 0$, then $\payA{\fA}{\fB} \leq \payA{\fAp}{\fB}$.
\end{enumerate}

In other words, the best strategy $f_p^*(\cdot)$ of player $p$ given the
strategy $f_o(\cdot)$ of the other player satisfies:
\begin{enumerate}
    \item $(f_p^*,\,f_o)$ is a complete-gold-coverage.
    \item $f_p^*(\cdot)$ uses the full capability of player $p$ up to $2M+1$
        line segments, i.e., $\seg(f_p^*)=\min(C_p,\,2M+1)$.
    \item If there is a strategy that starts at line 0 (i.e., $f_p(0)=0$) and
        satisfies both of the above constraints, then player $p$ plays such a
        strategy.
\end{enumerate}

Next we derive the \captrfn{} for the different
cases. We assume $\capB \leq \capA$ WLOG:
\begin{itemize}
    \item $2M+1 \leq \capB \leq \capA$: All resources are perfectly covered by
        both players. They receive the same payoff of $2M\rho$.
    \item $\capB < 2M+1 \leq \capA$: A perfectly covers all resources. B starts
        at $\fB(0)=0$ and uses all their capability.
        \begin{align*}
            \nrg(\fB) &=  M + \idiv{\capB - 1}{2} \\
            \payAsym &= 2M + (\rho - 1)\nrg(\fB) \\
                &= (\rho+1)M + (\rho - 1)\idiv{\capB - 1}{2} \\
            \payBsym &= \left(M + \idiv{\capB - 1}{2}\right) \rho +
                \left(M - \idiv{\capB}{2}\right) \mu \\
                &= \idiv{\capB - 1}{2}\rho - \idiv{\capB}{2}\mu + (\rho + \mu)M
        \end{align*}
    \item $\capB \leq \capA < 2M+1$: Let B first play an arbitrary strategy
        $\fB(\cdot)$. If $\fB(0)=0$, A will set $\fA(0)=1$ to ensure a
        complete-gold-coverage; otherwise if $\fB(0)=1$, A will set $\fA(0)=0$
        due to the second observation noted above. A can derive one of their
        best response $\fA(\cdot)$ according to the proof of
        \cref{lemma:agmg:cover-all-gold}. Following a similar reasoning from B's
        perspective, it can be shown that $\fB(\cdot)$ is also a best response
        of B given A's strategy $\fA(\cdot)$. Therefore, $(\fA,\,\fB)$ is a Nash
        equilibrium. There are two different classes of Nash equilibria: one
        with $\fA(0)=0$ and $\fB(0)=1$, and the other with $\fAp(0)=1$ and
        $\fBp(0)=0$. Let $\fA(0)=t$ and $\fB(0)=1-t$. We have:
        \begin{align*}
            \payAsym &= \idiv{\capA+t-1}{2}\rho - \idiv{\capA - t}{2}\mu
                +\idiv{\capB-t}{2}(\rho-1) + (\mu+1)M \\
            \payBsym &= \idiv{\capB-t}{2}\rho -  \idiv{\capB+t-1}{2}\mu
                +\idiv{\capA+t-1}{2}(\rho-1) + (\mu+1)M \\
            t &\in \{0,\,1\}
        \end{align*}
\end{itemize}
\Cref{thm:agmg} summarizes the results of these three cases.


%% file: fig/agmg-illu.tex
\begin{tikzpicture}[
    dot/.style = {minimum width=2em, minimum height=2em}
]

\definecolor{linecolor}{rgb}{0.8,0.8,0.8}
\definecolor{funccolorA}{rgb}{1.0, 0.0, 0.0} 
\definecolor{funccolorB}{rgb}{0.0, 0, 1.0} 

\draw[linecolor, ultra thin] (-1, 0) -- (7.5, 0);
\draw[linecolor, ultra thin] (-1, 0.7) -- (7.5, 0.7);

\foreach \i in {0,1,2.5} {
    \filldraw[black] (\i*2,     0.7) circle (2pt) node{};
    \filldraw[black] (\i*2+0.5, 0) circle (2pt) node{};
    \draw (\i*2+1,   0.7) node[cross=2pt] {};
    \draw (\i*2+1.5, 0) node[cross=2pt] {};
}

\node[dot] at (4.2, 0.35) {$\dots$};

\draw[funccolorA, dashed, thick] (-0.5, 0.7) -- (7, 0.7);
\draw[funccolorB, dashed, thick] (-0.5, 0.0) -- (6.25, 0.0) --
            (6.25, 0.7) -- (7, 0.7);

\end{tikzpicture}

%% file: fig/agmg-baseunit.tex
\pgfmathsetmacro{\xsep}{0.5}
\pgfmathsetmacro{\ysep}{0.4}
\pgfmathsetmacro{\yofhrule}{\ysep+0.2}
\pgfmathsetmacro{\yoftitle}{\yofhrule+0.2}

\definecolor{linecolor}{rgb}{0.8,0.8,1.0}
\tikzset{line/.style={linecolor, dashed, thick}}

\foreach [
    evaluate=\i as \x using {\i*\xsep},
    evaluate=\i as \y using {Mod(\i+1, 2)*\ysep},
    evaluate=\i as \t using {int(Mod(\i, 4)/2)}
] \i in {0, ..., 4} {
    \ifthenelse{\t=1}{\draw (\x, \y) node[cross=2pt] {}}
    {\filldraw[black] (\x, \y) circle (2pt) node {}};
}

\draw[ultra thin] (-0.5, \yofhrule) -- (4 * \xsep+0.5, \yofhrule);

\node at (1 * \xsep, -0.2) {\scriptsize $r_0$};
\node at (4 * \xsep, \ysep-0.2) {\scriptsize $r_1$};

\foreach [evaluate=\i as \t using {int(Mod(\i, 4})] \i in {0, ..., 4} {
    \node at (\i * \xsep, \yoftitle) {\t};
}
